\documentclass[conference]{IEEEtran}
\IEEEoverridecommandlockouts
\usepackage{cite}
\usepackage{soul}
\usepackage[english]{babel}
\usepackage{blindtext}
\usepackage{lipsum}

\usepackage{subcaption}

\usepackage{setspace}
\usepackage{amsmath,amsfonts}
\usepackage{array}

\usepackage{textcomp}
\usepackage{stfloats}
\usepackage{url}
\usepackage{verbatim}
\usepackage{graphicx}
\usepackage{float}
\usepackage{cite}
\usepackage{amsmath,amssymb,amsfonts}
\usepackage{graphicx}
\usepackage{textcomp}
\usepackage{xcolor}
\def\BibTeX{{\rm B\kern-.05em{\sc i\kern-.025em b}\kern-.08em
    T\kern-.1667em\lower.7ex\hbox{E}\kern-.125emX}}
\usepackage[english]{babel}
\usepackage{algorithm}
\usepackage{algpseudocode}
\usepackage{amsthm}
\usepackage{amssymb}
\usepackage{amsmath}
\usepackage{amssymb}	
\usepackage{amsfonts}
\usepackage{graphicx}
\usepackage{xcolor}
\usepackage{color,soul}
\usepackage{hyperref} 
\usepackage{pgf,tikz,pgfplots} 
\usepackage{dsfont}
\newcommand{\bd}{\boldsymbol}
\usepackage{tikz} 
\usetikzlibrary{arrows}
\usepackage{cleveref}

\newtheorem*{remark}{Remark}

\theoremstyle{definition}

\newtheorem{prop}{Proposition}

\begin{document}	
	\title{Vulnerability of Distributed Inverter VAR Control in PV Distributed Energy System
}

\author{\IEEEauthorblockN{Bo Tu, Wen-Tai Li and Chau Yuen\\Singapore University of Technology and Design (SUTD)\\ 8 Somapah Road, Singapore 487372\\{\tt\small \{bo\_tu@, wentai\_li@, yuenchau@\}sutd.edu.sg}}
}

\maketitle

\begin{abstract}

	This work studies the potential vulnerability of distributed control schemes in smart grids. 
	To this end, we consider an optimal inverter VAR control problem within a PV integrated distribution network. 
	First, we formulate the centralized optimization problem considering the reactive power priority and further reformulate the problem into a distributed framework by an accelerated proximal projection method. 
	The inverter controller can curtail the PV output of each user by clamping the reactive power. 
	To illustrate the studied distributed control scheme that may be vulnerable due to the two-hop information communication pattern, we present a heuristic attack injecting false data during the information exchange. 
	Then we analyze the attack impact on the update procedure of critical parameters. 
	A case study with an eight-node test feeder demonstrates that adversaries can violate the constraints of distributed control scheme without being detected through simple attacks such as the proposed attack.

\end{abstract}


\section{Introduction}

As the penetration level of photovoltaic (PV) increases rapidly in the distribution grid, the intention of a normal energy consumer to sell surplus solar power back to the grid also boosts. 
Hence, numerous smart inverters appear surrounding the neighborhood \cite{lin2020research}.
Regarding the IEEE 1574-2018 Standard \cite{1547-2018}, smart inverters have to convert the photovoltaic output power as AC power and provide or absorb a specific amount of reactive power to level the fluctuation of nodal voltage.

Contrary to scenarios within a transmission network, the conventional centralized control method may not be applicable for the distribution network with large-scale inverters integrated because of the potential heavy burden on information exchange and computation.
Therefore, in recent years, inverter-based decentralized control schemes have been extensively studied to deal with nodal voltage regulation issues.
\cite{farivar2012optimal} initiates and explores the advanced voltage-ampere reactive (VAR) control capability of inverter interfaces to the grid.
In light of the above work, the decentralized optimization method has attracted great interest in the area of inverter voltage control within the distribution network 
\cite{zhu2015fast, qu2019optimal,lin2016decentralized}.
All the studies above did not investigate the vulnerability of their decentralized controllers, namely, the security of the parameter iteration. 
In the real world application, a so-called reactive power priority principle \cite{1547-2018}, aiming to guarantee the stability of the nodal voltage, may hamper the efficiency of the renewable energy supply and motivate PV owners to conduct misbehavior in the system.

The main contributions of this work are two parts.
First, we characterize an accelerated proximal gradient method for a distributed control scheme considering the limitation of the rated apparent power.
Second, we investigate the vulnerability of the local information update mechanism of the inverter control scheme.
Specifically, we demonstrate that even a simple adversarial behavior, namely, a heuristic attack, can deviate the distributed controller to converge to a wrong value; hence, the attacker can feed more PV power to the grid and bypass the reactive power priority.

The notation of this paper is addressed as follows. 
The symbol $:=$ denotes the equality by definition.
We denote the set of real numbers as $\mathbb{R}$ and non-negative real numbers as $\mathbb{R}_+$.
The bold front indicates the vector.
The symbol $\bd{1}_T$ represents the $T$ dimensional all-ones vector, and the symbol $\bd{A}^{\mathsf T}$ represents the transpose of a vector $\bd{A}$.
The symbol ${\rm{I}}_T$ represents the $T$ dimensional identity matrix.

\begin{figure}[b!]
	\centering
	\includegraphics[scale=0.35]{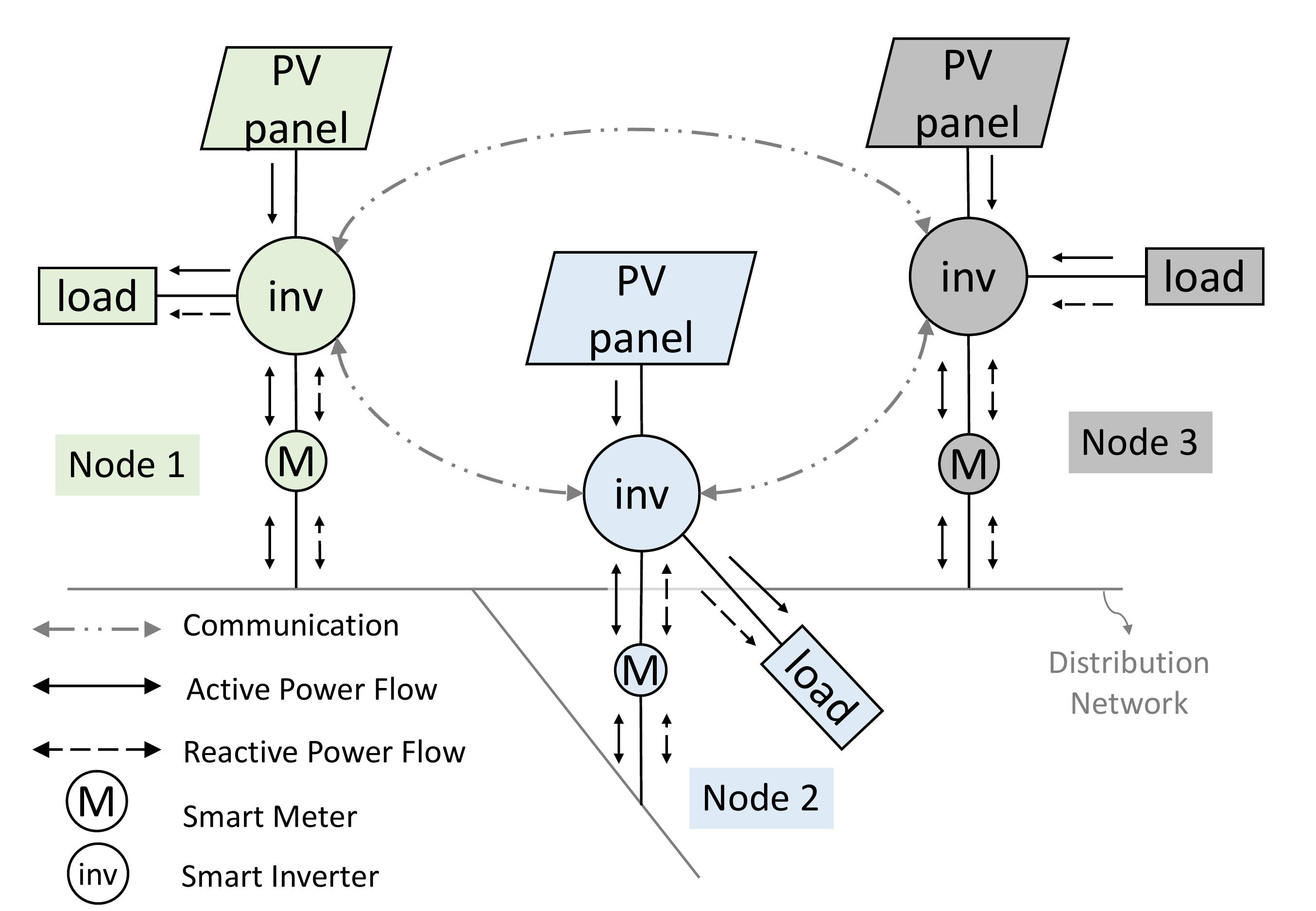}
	\caption{Distribution network with smart inverter integrated}
	\label{fig:system} 
\end{figure}

\section{Problem Formulation}

\subsection{System Preliminaries}

This paper considers a single-phase radial distribution feeder system that comprises multiple PV integrated residents.
Such a system can be denoted as a graph $G(\tilde{\mathcal{N}},\mathcal{E})$ consisting of a set of $N+1$ nodes $\tilde{\mathcal{N}} := \{0,1,\ldots,N\}$,
and a set of links $\mathcal{E}\subseteq\{\{i,j\}|i,j\in\tilde{\mathcal{N}},~i\neq j\}$.
Removing the substation node with index $0$, we have the interest set $\mathcal{N} :=\tilde{\mathcal{N}}/\{0\}$.
For node $i\in\mathcal{N}$, $v_i$ denotes the voltage magnitude, $p_i$ denotes the node active power injection, and $q_i$ denotes the reactive power injection.
For each link $(i,j)\in\mathcal{E}$, we denote by $x_{ij}$ the reactance value and $r_{ij}$ the resistance value.
The term $P_{ij}$ and $Q_{ij}$ denote the real and reactive power from node $i$ to $j$, respectively.
As shown in Fig. \ref{fig:system}, each node equips with a PV panel, smart inverter, and a smart meter.
The controller-based smart inverter can communicate with each other.
The active and reactive power can be injected into or absorbed from the network.
Due to the smart meter, each node will be charged or paid for using or providing active power.


A smart inverter needs to provide reactive power, positive or negative, as required in IEEE 1547-2018 \cite{1547-2018}.
Practically, the active and reactive power have to satisfy an inequality constraint as follows:
\begin{align}
	\label{eq:inverter}|q_i| \leq \sqrt{\bar{s}_i^2 - \tilde{p}_i^2},~ i = 1,\ldots,N,
\end{align}
where $\bar{s}_i$ is the upper bound of the apparent power capacity, a given value of a specific smart inverter of node $i$.
The reactive power ${q}_i$ is the controllable variable by the smart inverter, either positive (injection) or negative (absorption).
Furthermore, $p_i$ and $\tilde{p}_i$ are obtainable due to smart inverters and smart meters in our configuration.

\subsection{Linearized Branch Flow Model for Radial Networks}

We adopt a simplified Distflow model \cite{baran1989optimal}, which is extensively used in voltage control problems in distribution systems \cite{farivar2012optimal}.
For each link $(i,j)\in\mathcal{E}$, by ignoring the power loss term, we can derive the linearized Distflow model as \cite{zhu2015fast} 
\begin{subequations}
\begin{align}
	P_{ij} &= \sum_{k:(j,k)\in\mathcal{E}}P_{jk} - p_j,\\
	Q_{ij} &= \sum_{k:(j,k)\in\mathcal{E}}Q_{jk} - q_j,\\
	v_i - v_j &= 2(r_{ij}P_{ij} + x_{ij}Q_{ij}).
\end{align}
\end{subequations}

By denoting the voltage vector as $\bd{v} := [v_1,\ldots,v_N]^\mathsf{T}\in\mathbb{R}^N_+$, active power injection vector as $\bd{p} := [p_1,\ldots,p_N]^\mathsf{T}\in\mathbb{R}^N$, and reactive power injection vector as $\bd{q} := [q_1,\ldots,q_N]^\mathsf{T}\in\mathbb{R}^N$, we can derive the following equation: 
\begin{align}
	\label{eq:grid}
	\bd{v} = \bd{Rp} + \bd{Xq} + v_0\mathbf{1}_N,
\end{align}
where the resistance information matrix $\bd{R}\in\mathbb{R}^{N\times N}$ can be denoted as $\bd{R} := 2(-\bd{A}^{\sf -1}){\rm diag}(\bd{r})(-\bd{A}^{\sf -1})^{\sf T}$ \cite{kekatos2015voltage}\cite{kekatos2015fast},
where matrix $\bd{A}$ denotes an incidence matrix without the first column of node $0$, and $\bd{A}$ is an $N$-dimensional square matrix for a radial structure network;
besides, $\rm{diag}(\bd{r})$ denotes a diagonal matrix with its diagonal entry $r_{nn}$ equivalent to the corresponding $r_n$, for $n = 1,\ldots, N$;
likewise, the reactance information matrix $\bd{X}\in\mathbb{R}^{N\times N}$ can be denoted as $\bd{X} := 2(-\bd{A}^{\sf -1}){\rm diag}(\bd{x})(-\bd{A}^{\sf -1})^{\sf T}$.

\subsection{Centralized Inverter Var Control}

Accordingly, we formulate the inverter VAR control problem as follows:
\begin{equation}
\begin{aligned}
	\label{eq:objV}\min & \quad\frac{1}{2}||\bd{v} -\bd{\mu}||_2^2\\
	{\text{s.t.}}& \quad (\ref{eq:inverter})~\rm{and}~(\ref{eq:grid})~{\text{hold}}
\end{aligned}
\end{equation}
where $\bd{\mu}=[\mu_1,\ldots,\mu_N]^{\mathsf T}$, with $\mu_i =220V,~i\in\mathcal{N}$ denoting the desired voltage value.

Substituting (\ref{eq:grid}) into the objective function of (\ref{eq:objV}), and rewriting the inequality constraints in a compact manner, we can recast the problem (\ref{eq:objV}) as follows:
\begin{align}
	\label{eq:objrecast}
	\min_{\bd{q}} \quad &\frac{1}{2} ||\bd{Rp} + \bd{Xq}||_2^2\\\nonumber
	{\rm s.t.} \quad &\bd{C}\bd{q} - \bd{D}\bd{Q} \leq \bd{0}
\end{align}
where $\bd{C} := [
{\bd{\rm I}}_N,~-{\bd{\rm I}}_N
]^{\mathsf T}$, $\bd{D} := [
{\bd{\rm I}}_N,~{\bd{\rm I}}_N
]^{\mathsf T}$, $\bd{Q} :=[Q_1,\ldots,Q_N]^{\mathsf T}$, and $Q_j := \sqrt{\bar{s}_j^2 - \tilde{p} _j^2}$, $j=1,\ldots,N$.

Such a problem can be solved easily via existing solvers in a centralized manner because of the convexity of the decision variable $\bd{q}$. 
However, as for a distribution feeder network, the centralization-like operation also indicates a higher requirement on the communication bandwidth and computation capability, which may increase retailers' costs.
Therefore, in the sequel, the centralized control issue will be reformulated into a decentralized way.


The Lagrange duality approach is extensively employed to solve the constrained optimization problem  (\ref{eq:objrecast}) \cite{boyd2004convex}.
Since the derived Lagrange dual problem is convex \cite{boyd2004convex}, it has a convenient form to implement the proximal algorithms, such as 
the proximal gradient method and the Alternating direction method of multipliers  \cite{parikh2014proximal}.
Both of the above methods solve the primal problem in a decentralized way.
This work focus on the accelerated proximal gradient method, which is also well-known as FISTA \cite{beck2009fast}.


Accordingly, the Lagrangian $\mathcal{L}:\mathbb{R}^N\times\mathbb{R}^{N}\times\mathbb{R}^{N} \to \mathbb{R}$ of the above issue is given as
\begin{align}
	\label{eq:lagrange}
	\mathcal{L}(\bd{q},\bar{\bd{\lambda}},\underline{\bd{\lambda}}) = &\frac{1}{2}(\bd{Rp}+\bd{Xq})^{\mathsf T}(\bd{Rp}+\bd{Xq}) +\bar{\bd{\lambda}}^{\mathsf T}(\bd{q}-\bd{Q})\nonumber\\
	 &+\underline{\bd{\lambda}}^{\mathsf T}(-\bd{q}-\bd{Q}),
\end{align}
where $\bar{\bd{\lambda}}$, $\underline{\bd{\lambda}}\in \mathbb{R}^N$ are the Lagrange multipliers.
For simplicity, we write $\bd{\lambda}:=[\bar{\bd{\lambda}}^{\mathsf T},\underline{\bd{\lambda}}^{\mathsf T}]^{\mathsf T}\in\mathbb{R}^{2N}$ in a compact way.
The dual function $\mathcal{G}:\mathbb{R}^{2N}\to\mathbb{R}$ is given as $\mathcal{G}(\bd{\lambda}) := \inf_{\bd{q}\in\mathcal Q}\mathcal{L}(\bd{q},\bd{\lambda})$, where $\mathcal{Q}$ is the feasible region of the decision variable $\bd{q}$.
Since $\mathcal{L}$ is a convex quadratic function on $\bd{q}$, we can derive the first-order optimal condition $\nabla_{\bd{q}} \mathcal{L}(\bd{q},\bd{\lambda}) = \bd{0}$,
therefore, we have $\bd{q} = -\bd{X}^{-2}\bd{C}^{\mathsf T}\bd{\lambda} - \bd{X}^{\mathsf{-1}}\bd{R}\bd{p}$.
By assuming the network is homogeneous, it is obvious that $\bd{X}^{\mathsf{-1}}\bd{R}$ can be derived as $\xi\mathbf{I}_{N}$, where $\xi\in\mathbb{R}$ is a scalar.
Therefore, by denoting $\bd{B}:=\bd{X}^{-1}$, we have 
	$\bd{q} = -\bd{B}^2\bd{C}^{\mathsf T}\bd{\lambda}\ -\xi\bd{p}$,
then let us substitute $\bd{q}$ into (\ref{eq:lagrange}) to derive the objective function of the Lagrange dual problem as
\begin{align}
\label{eq:Ldual}
	\tilde{\mathcal{G}}(\bd{\lambda}) := -\frac{1}{2} \bd{\lambda}^{\mathsf T}\bd{CB}^2\bd{C}^{\mathsf T}\bd{\lambda} - \xi\bd{\lambda}^{\mathsf T}\bd{Cp} - \bd{\lambda}^{\mathsf T}\bd{C}\bd{q},
\end{align}
hence, by defining $\mathcal{G}:= -\tilde{\mathcal{G}}$, the Lagrange dual problem can be denoted as
\begin{align}
\label{eq:dualproblem}
	\min_{\bd{\lambda}\geq \bd{0}} ~ \mathcal{G}(\bd{\lambda}) ,
\end{align}
where the Lagrange multiplier $\bd{\lambda} \geq \bd{0}$ is equivalent to a case $\bd{\lambda}$ lies within a non-negative orthant within a Cartesian coordinate plane.
Note that the dual problem is convex \cite{boyd2004convex}.

\subsection{Accelerated proximal gradient method}

So far, we have reformulated the primal problem to the associated dual problem.
Inspired by \cite{8452988}, we employ an accelerated proximal gradient method \cite{parikh2014proximal} to solve the problem iteratively.
Accordingly, problem (\ref{eq:dualproblem}) can be recast as an unconstrained optimization problem
\begin{equation}
	\label{eq:lagdual}
	\min \quad \mathcal{G}(\bd{\lambda}) + \mathcal{I}_\mathcal{C}(\bd{\lambda}),
\end{equation}
where $\bd{\lambda}\in\mathbb{R}^{2N}$ is the decision variable, the function $\mathcal{I}_\mathcal{C}:\mathbb{R}^{2N}\to\mathbb{R}\cup \{+\infty\}$ is an indicator function of a convex set $\mathcal{C}\subseteq\mathbb{R}_+$ given as follows
\begin{equation}
	\label{eq:indicate}
	\mathcal{I}_\mathcal{C} := 
	\begin{cases}
	0&{\rm if}~{\lambda_{i}}\in \mathcal{C}\\
	+\infty & {\rm if}~{\lambda_{i}}\notin \mathcal{C}
	\end{cases},~i=1,\ldots,2N.
\end{equation}
Let us leverage the accelerated proximal gradient method \cite{beck2009fast} to solve the reformulated Lagrange dual problem (\ref{eq:lagdual}).
Hence, we have the parameter update rules as follows
\begin{align}
	\label{eq:update3}\bd{\lambda}(k) &:= {\textbf{prox}}_{\alpha\mathcal{I}_\mathcal{C}}\Big(\bd{\theta}(k) - \alpha\nabla \mathcal{G}\big(\bd{\theta}(k)\big)\Big),\\
	\label{eq:update2}\omega(k+1) &:= \frac{1+ \sqrt{1+4\omega(k)^2}}{2},\\
	\label{eq:update1}\bd{\theta}(k+1) &:= \bd{\lambda}(k) + \frac{\omega(k)-1}{\omega(k+1)}\big(\bd{\lambda}(k) - \bd{\lambda}(k-1)\big),
\end{align}
where {$\bd{\theta}:=[\bar{\bd{\theta}}^{\mathsf T},\underline{\bd{\theta}}^{\mathsf T}]^{\mathsf T}\in\mathbb{R}^{2N}$}, $\alpha>0$ is a fixed step size.

\section{Security Issue Analysis}

The decentralized control scheme offers a flexible and efficient way to adjust the reactive power to keep the stability of the nodal voltage.
However, such a scheme may be vulnerable to some adversarial behavior.
In the sequel, we analyze the vulnerability of the information local-update mechanism of the very controller.
Then we introduce a reactive power priority principle that may motivate a PV owner to conduct some misbehaviors.
At last, demonstrate that even a heuristic attack can mislead the controller to converge to the wrong value.

\subsection{Parameter update rule analysis}


As shown in the previous section, the update rule (\ref{eq:update3}) is based on a proximal operator, which is not easy to analyze.
In this subsection, an equivalent form is derived, whereby the underlying communication mechanism can be unveiled.

	\begin{prop}
		\label{prop1}
		The parameter update rule (\ref{eq:update3}) is equivalent to 
		\begin{align}
			\lambda_i(k) = 
			\begin{cases}
				h_i(k)&,~{\rm if }~h_i(k)\geq 0\\
				0&,~{\rm if }~h_i(k)< 0
			\end{cases}, ~i = 1,\ldots, 2N
		\end{align} 
	where $\bd{h}(k) := \bd{\theta}(k) - \alpha\nabla \mathcal{G}(\bd{\theta}(k))\in\mathbb{R}^{2N}$. 
	\end{prop}
	\begin{proof}
		Since the proximal operator is given as 
		\begin{equation}
			\label{eq:lambda}
			\bd{\lambda}(k) = \arg \min_{\bd{\lambda}\in \mathcal{C}} \Bigg\{\mathcal{I}_C(\bd{\lambda}) + \frac{1}{2\alpha}||\bd{\lambda} - \bd{h}(k)||_2^2\Bigg\},
		\end{equation}
		by recalling (\ref{eq:indicate}), $\bd{\lambda}$ can be derived as follows
		\begin{equation}
			\label{eq:lambda_projection}
			\bd{\lambda}(k) = \arg \min_{\bd{\lambda}\in \mathcal{C}} \frac{1}{2\alpha}||\bd{\lambda} - \bd{h}(k)||_2^2.
		\end{equation}
		Since the set $\mathcal{C}\subseteq\mathbb{R}^{2N}_+$, in which case (\ref{eq:lambda_projection}) indicates the projection onto the positive axis of $\bd{\lambda}$.
	\end{proof}

Therefore, it is easy to understand the decentralized iterative operator (\ref{eq:update3}) in a distributed manner, i.e., the parameter can be decided by sharing the information with each other rather than conducted by a centralized operator.

By defining $\tilde{\bd{B}} := {\bd{B}}^2$, let us recall (\ref{eq:Ldual}) and (\ref{eq:update3}) to derive the term $\nabla \mathcal{G}$ as
\begin{align}
	\nabla \mathcal{G}(\bd{\theta}) &= 
	\begin{pmatrix}
		\tilde{\bd{B}}&-\tilde{\bd{B}}\\
		-\tilde{\bd{B}}&\tilde{\bd{B}}
	\end{pmatrix} 
	\begin{pmatrix}
		\bar{\bd{\theta}}\\
		\underline{\bd{\theta}}
	\end{pmatrix} + \xi\bd{Cp} + \bd{C}\bd{q}\\
	&= \begin{pmatrix}
		\tilde{\bd{B}}(\bar{\bd{\theta}}-\underline{\bd{\theta}})\nonumber\\
		-\tilde{\bd{B}}(\bar{\bd{\theta}}-\underline{\bd{\theta}})
	\end{pmatrix} 
	+ \xi\bd{Cp} + \bd{C}\bd{q}.
\end{align}
We can see that the term $\tilde{\bd{B}}$ is the coefficient of the term $\bd{\theta}$ to decide the update of $\bd{\lambda}$.
We further rewrite the non-zero part of (\ref{eq:update3}) in a distributed manner: 
\begin{align}
	\label{eq:updatei1}\bar{\lambda}_i(k) = & \bar{\theta}_i(k) -\alpha(\tilde{B}_{ri}(\begin{pmatrix}\bar{\theta}_1(k)\\\vdots\\\bar{\theta}_N(k)\end{pmatrix}
	-\begin{pmatrix}\underline{\theta}_1(k)\nonumber\\\vdots\\\underline{\theta}_N(k)\end{pmatrix}) \\&+\xi p_i(k)+\bar{q}_i(k)),
\end{align}
\begin{align}
	\label{eq:updatei2}\underline{\lambda}_i(k) = & \underline{\theta}_i(k) +\alpha(\tilde{B}_{ri}(\begin{pmatrix}\bar{\theta}_1(k)\\\vdots\\\bar{\theta}_N(k)\end{pmatrix}
	-\begin{pmatrix}\underline{\theta}_1(k)\nonumber\\\vdots\\\underline{\theta}_N(k)\end{pmatrix}) \\&+ \xi p_i(k)+\underline{q}_i(k)),
\end{align}
where $\tilde{\bd{B}}_{ri}$ denotes the $i$th row of the matrix $\tilde{\bd{B}}$.
Likewise, $\omega_i$ and $\bd{\theta}_i$ can also be written in a distributed manner.

\subsection{Local information update}

One can observe that the terms $\bar{\lambda}_i$ and $\underline{\lambda}_i$ presented in (\ref{eq:update3}) also depend on $\bar{\theta}_j$ and $\underline{\theta}_j$ of the other node $j\in\mathcal{N}\setminus\{i\}$, and parameters $\theta_i$ and $\omega_i$ update only depend on the historical data of node $i$ itself.
Given such a paradigm, each node has to broadcast its own data and receive the other nodes' data to complete the $k$th iteration.
The real-time communication burden would be an issue.
However, as for the radial distribution network, the controller spontaneously works in a distributed manner.
It also indicates that the controller may update without massive information exchange but locally sharing.

Let us investigate matrices $\bd{X}$ and $\bd{R}$ for further understanding.
Recall that $\bd{X}\in\mathbb{R}^{N\times N}$ is the reactance square matrix of the particular power system.
The term $\tilde{\bd{B}} = ({X^{-1}})^2$ can be regarded as a two-hop information matrix \cite{lynch1996distributed}, partially indicating the two-layer nearest neighborhood of each node.
Note that such a characteristic can be easily derived via the incidence and adjacent matrix of the interest network.
Therefore, for each row of the matrix $\tilde{\bd{B}}$, with respect to (\ref{eq:updatei1}) and (\ref{eq:updatei2}), we can obtain a reduced correlations scale, i.e., each node within the network, only exchange message with its two-layer nearest neighbors rather than using the full bandwidth to communicate with all the other nodes.
Such a procedure is also well known as two-hop communication \cite{lynch1996distributed}.

\begin{prop}
	Given the radial network $G$, the matrix $\tilde{\bd{B}}$ indicates the two-hop neighbor of each node.
\end{prop} 
\begin{proof}
	Without loss of generality, we consider a four-layer binary tree rooted by node $m_1$.
	Practically, a node $m_0$ representing the substation is the father node of $m_1$.
	By eliminating the column of $m_0$ of the associated incidence matrix, a sparsity square matrix can be derived.
	Then, it is easy to observe that for the symmetric matrix $({(-A^{\mathsf -1})(-A^{\mathsf -1})^{{\mathsf T}}})^{\mathsf -2}$, the two-hop region of each node is non-zero, and nodes outside the region are all zero.
\end{proof}

Therefore, to fulfill the control scheme in a distributed manner, the controller only shares information with his two-hop neighbors. 
The advantages are the limited communication bandwidth consumption and less central computing burden.
However, such a scheme diluting the center's power, such as the bad data detection, may leave potential operational space for some adversarial behaviors.

\subsection{Reactive power priority}

Recalling the constraint (\ref{eq:inverter}), the active, reactive, and apparent power relationship can be illustrated inside a Cartesian coordinate system with a half-circle of radius $\bar{s}_i$ centered at the origin, as shown in Fig. \ref{fig:curtail}.
\begin{figure}[htb!]
	\centering
	\includegraphics[scale=0.45]{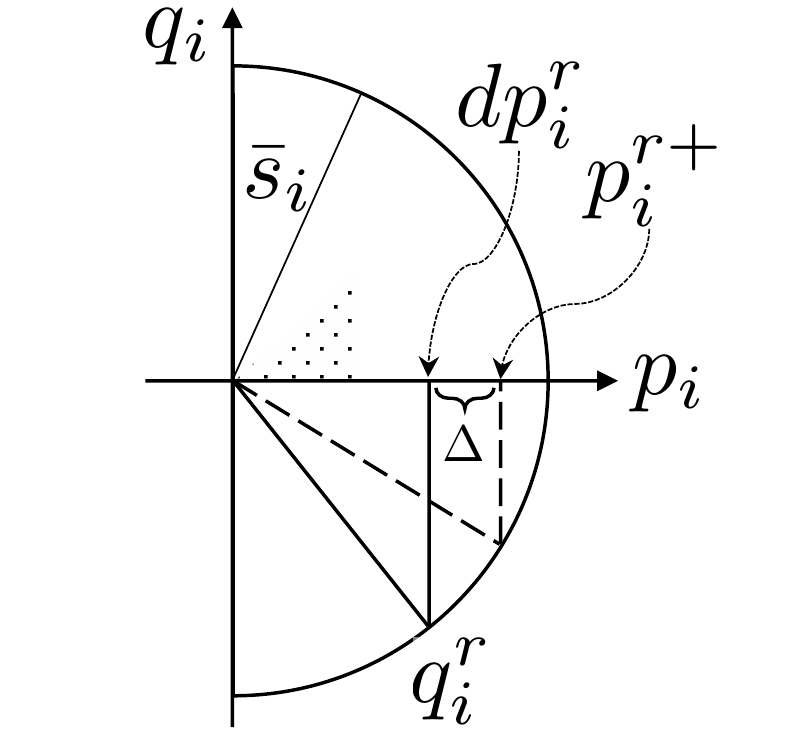}
	\caption{Operation orthant of the smart inverter }
	\label{fig:curtail} 
\end{figure}
As for the small active power cases, the corresponding reactive power, which is less than the upper bound, is enough to stabilize the nodal voltage.
As the active power reaches the rated value or some particular proportionality, the reactive power reaches $q_i^r$.
Once $p_i^{r+} \geq d p_i^r$, where $d$ is a coefficient, the surplus power $\Delta$ will be curtailed.

\begin{remark}
	Such control logic can be regarded as the principal of reactive power priority, where the capacity of the smart inverter is used for clamping the reactive power to curtail the active power.
\end{remark}

According to the California Public Utilities Commission Resolution E-4920 \footnote{\url{https://www.pge.com/tariffs/tm2/pdf/ELEC_4920-E-A.pdf}}, all the inverters joining the network are required to comply with the reactive power priority principle.
However, to the best of our knowledge, all the proposed distributed control schemes have to curtail a particular amount of active power if the PV power is surplus a lot.
As a result, the principle may cause the loss of the opportunity to provide active power.
From the viewpoint of a PV-based node that intends to provide more renewable energy to the grid may be restrained by the particular mechanism.
Such a curtailment suppresses the potential benefit of the end-point user, which may motivate users to bypass the restraint via multiple misbehaviors.

\subsection{Heuristic attack}

Generally, there are three categories of misbehavior such as node's crash, malicious, and Byzantine modes \cite{leblanc2013resilient}.
In this work, we focus on the malicious misbehaving, where an attacker broadcasts the same malicious message to all of his two-hop neighbors.
We assume that the only attacker, node $j\in\mathcal{N}$, can compromise his own data but cannot modify the data passed from the other normal nodes, which can be realized by an authenticated technique \cite{he2013sats}.

The parameter $\theta_i^a(k)$ under malicious misbehavior is presented as
\begin{align}
\theta_i^a(k) := \theta_i(k) + a,
\end{align}
where $a\in\mathbb{R}$ is the manipulated term keeping identical during each iteration.
Such malicious misbehavior is also known as the constant injection attack\cite{he2013sats} or a heuristic attack.
Then the broadcast with contaminated data will influence the control performance of the other nodes, following the parameter update and share fashion given in \cref{eq:updatei1,eq:updatei2}.

In this work, let us consider only one attacker within the network.
Such a consideration is reasonable since, in the distribution network application, normal end-point energy users may not be able to collude and cooperate to commit misbehavior.
On the other hand, the cooperative attack is out of the scope of this work.

Note that such an attack can also be regarded as a false data injection attack (FDIA) \cite{liu2011false}.
However, the conventional FDIA mainly focuses on how to manipulate the residual within the state estimation procedure.
Although with a similar formation to FDIA, the attack in this work impacts the running of the distributed controller, which is different.
Therefore, the novelty of this work does not stem from the attack's formation but the underlying working mechanism.
To eliminate the ambiguity, we avoid referring to the proposed attack as FDIA.

Let us consider a three-node toy example, as shown in fig.\ref{fig:information}.
The upper subplot represents the $k_{th}$ iteration.
The attacker, node 1, sends tampered data $\theta_1^a(k) $ to its neighbors; then, in the $(k+1)_{th}$ iteration , the lower subplot, due to the contaminated parameter $\theta_1^a(k)$, node 2 and 3 derive $\lambda_2^{a_1}(k+1),~\theta_2^{a_1}(k+1)$, and terms $\lambda_3^{a_1}(k+1),~\theta_3^{a_1}(k+1)$ transmitting to node 1.
As the iterations continue, the malicious behavior of the attacker is propagated throughout the network.

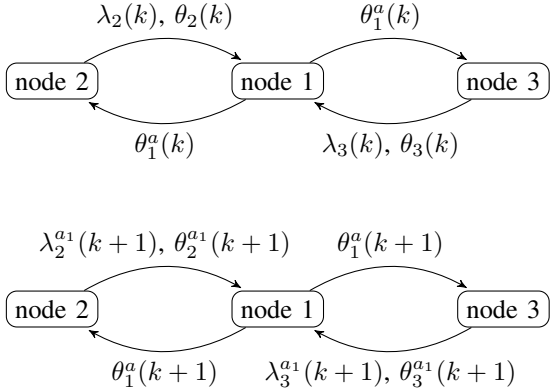
\begin{figure}[tbh!]
	\begin{minipage}[b]{\paperwidth}
\hspace{0.7cm}
		\begin{tikzpicture}[->,>=stealth',shorten >=2pt,node distance=1cm,auto,main node/.style={rectangle,rounded corners,draw,align=center}, scale=1, every node/.style={scale=1}]
			\node[main node] (1) {node 1};
			\node[main node] (2) [xshift=-2cm][left of=1] {node 2};
			\node[main node] (3) [xshift=2cm][right of=1] {node 3};
			\path
			(1) edge [bend left=30] node {{$\theta_1^a(k)$}} (3)
			(3) edge [bend left=30] node {$\lambda_3(k)$, $\theta_3(k)$} (1)
			(1) edge [bend left=30] node {{$\theta_1^a(k)$}} (2)
			(2) edge [bend left=30] node {$\lambda_2(k)$, $\theta_2(k)$} (1);
		\end{tikzpicture}
\end{minipage}\\
\vspace{0.3cm}\\
\begin{minipage}[b]{\paperwidth}
\hspace{0.7cm}
		\begin{tikzpicture}[->,>=stealth',shorten >=2pt,node distance=1cm,auto,main node/.style={rectangle,rounded corners,draw,align=center},scale=1, every node/.style={scale=1}]
			\node[main node] (1) {{node 1}};
			\node[main node] (2) [xshift=-2cm][left of=1] {{node 2}};
			\node[main node] (3) [xshift=2cm][right of=1] {{node 3}};
			\path
			(1) edge [bend left=30] node {{$\theta_1^a(k+1)$}} (3)
			(3) edge [bend left=30] node {{$\lambda_3^{a_1}(k+1)$, $\theta_3^{a_1}(k+1)$}} (1)
			(1) edge [bend left=30] node {{$\theta_1^a(k+1)$}} (2)
			(2) edge [bend left=30] node {{$\lambda_2^{a_1}(k+1)$, $\theta_2^{a_1}(k+1)$}} (1);
		\end{tikzpicture}
	\end{minipage}
	\caption{Iteration $k$ (upper) and $k+1$ (lower)}
	\label{fig:information}
\end{figure}

Once the tampered parameter $\bd{\lambda}^a$ is settled by a proper stop criterion, substituting $\lambda_i^a$ into $\bd{q}$, the smart inverter output of node $i$, i.e., the attacker, is denoted as follows
\begin{align}
\label{eq:reactive}
	{q}_i^a = -\tilde{\bd{B}}_{ri}^2(\bar{\bd{\lambda}}^a - \underline{\bd{\lambda}}^a) -\xi p_i.
\end{align}

This section discusses a distributed inverter control scheme that is vulnerable because the local information exchange mechanism lacks centralized bad data detection.
Motivated by the potential curtailment of PV active power, a heuristic attack can be implemented to influence the controller.

\section{Case study}

\subsection{System configuration}

We demonstrate the analysis on a single-phase eight-node homogeneous radial distribution system as shown in Fig. \ref{fig:linefigure}, where node $N_0$ represents the substation of the network.
Node 3 has two branches, and each branch has two children nodes.
Let us assume links between any two nodes have identical physical characteristics, hence the same $x$ and $r$. 

\begin{figure}[hbt!]
\begin{center}
\begin{tikzpicture}[node distance={15mm}, thin, main/.style = {draw, circle}, scale=0.9, every node/.style={scale=0.9}] 
\node[main] (1) {$N_1$}; 
\node[main] (2) [right of=1] {$N_2$};
\node[main] (3) [right of=2] {$N_3$};
\node[main] (6) [below of=3] {$N_6$};
\node[main] (7) [right of=6] {$N_7$};
\node[main] (4) [right of=3] {$N_4$};
\node[main] (5) [right of=4] {$N_5$};
\node[main] (8) [left  of=1] {$N_0$};

\draw (1) -- (2);
\draw (2) -- (3);
\draw (3) -- (4);
\draw (4) -- (5);
\draw (3) -- (6);
\draw (6) -- (7);
\draw (1) -- (8);
\end{tikzpicture} 
\end{center}
\caption{8-node system demonstration}
\label{fig:linefigure}
\end{figure}
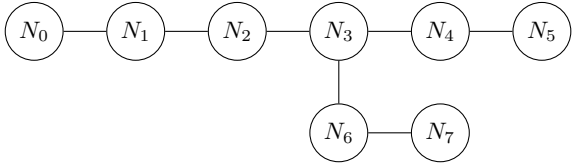

Given the homogeneous network assumption, each line between two nodes is 200 meters of copper wire, and the diameter of the wire is two millimeters.
Then, the reactance and resistance can be derived as $X = 0.4~\Omega$ and $R = 1.1 \Omega$, respectively.
The rated apparent power is given as 5 kVA for each smart inverter.
The active power is given as $p = $[3.6 -4 2.26 -2.5 4.85 3.31 2.43], the unit is kWh.
Besides, the desired stable voltage is 220V.

Let us demonstrate the effect of the heuristic attack where the attacker only needs to tamper with the attacker-side data to provide to his two-hop neighbors.
Given the active power reading $\bd{p}$, the controller uses (\ref{eq:reactive}) under an adversarial environment to derive the optimal reactive power.
In this paper, we assume that the iteration stops once the reactive power fluctuation $|q_i(k) - q_i(k+1)|\leq \tau$ continuing for 1000 iterations, where $\tau$ is a suitable small value.
Note that such a stop criterion may be flexible in the different application scenarios.
Once the criterion is satisfied, the controller controls the actuator to adjust the reactive power output.

\subsection{Numerical results}

\begin{figure}[tbh!]
\centering
	\includegraphics[scale = 0.43]{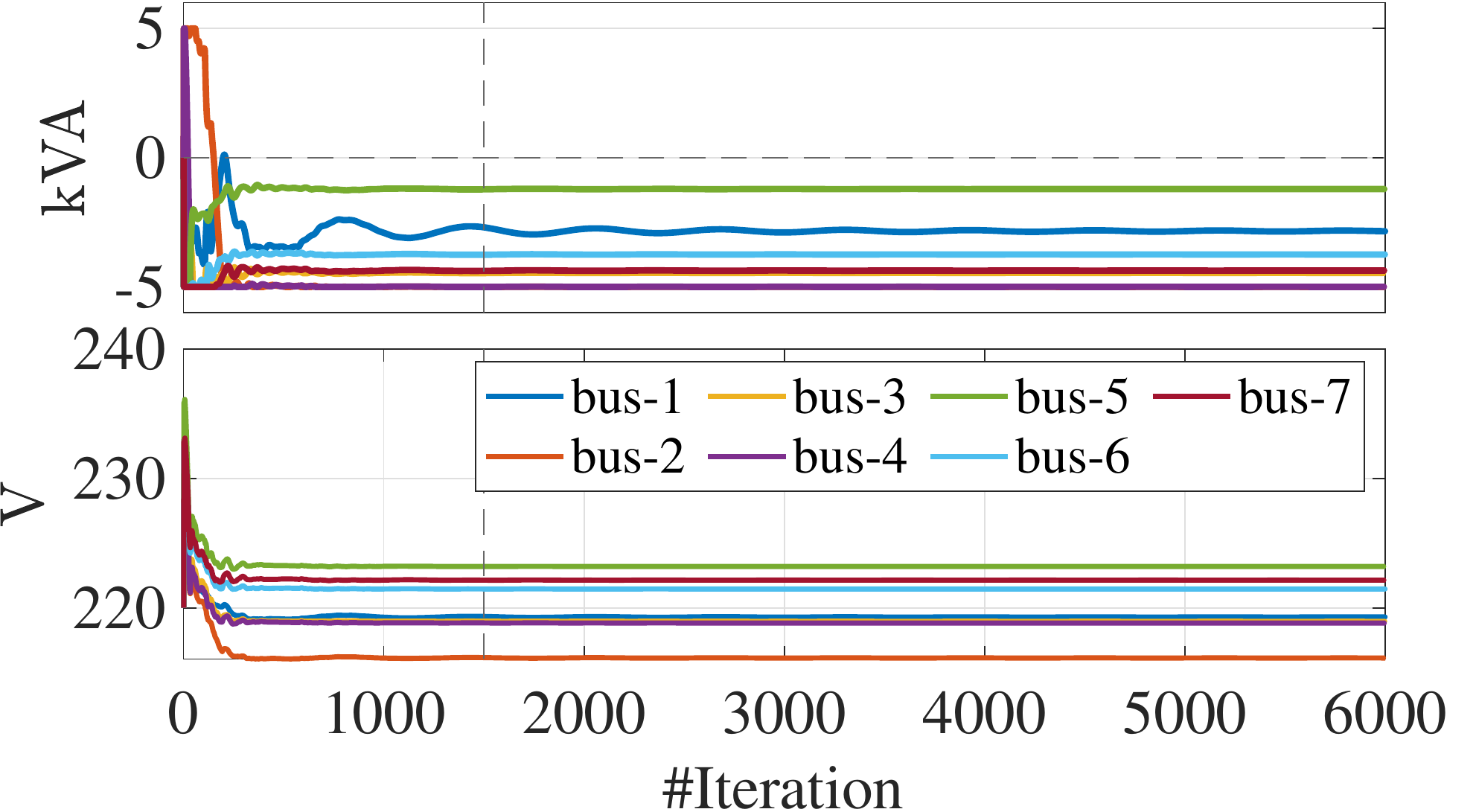}
\caption{Controller iterative convergence without attack. a) Reactive power; b) Voltage.}
\label{fig:bftera} 
\end{figure}
Firstly, let us demonstrate the validity of the distributed controller.
The value of the reactive power converges after 500 iterations, as shown in the upper subplot of Fig. \ref{fig:bftera}.
Likewise, the voltage value also converges and reaches a feasible region, as shown in the lower subplot.
Node 5 is forced to output at least -1 kVA reactive power to the grid to comply with the California Public Utilities Commission Resolution E-4920.
On the contrary, node 3 and node 4 are forced to output -5 kVA reactive power to keep the voltage stable within a feasible region, indicating that such two nodes can not feed PV power to the grid because of the apparent power limitation.
Such a result demonstrates the validity of the derived distributed control scheme.

\begin{figure}[t!]
\centering
	\includegraphics[scale = 0.43]{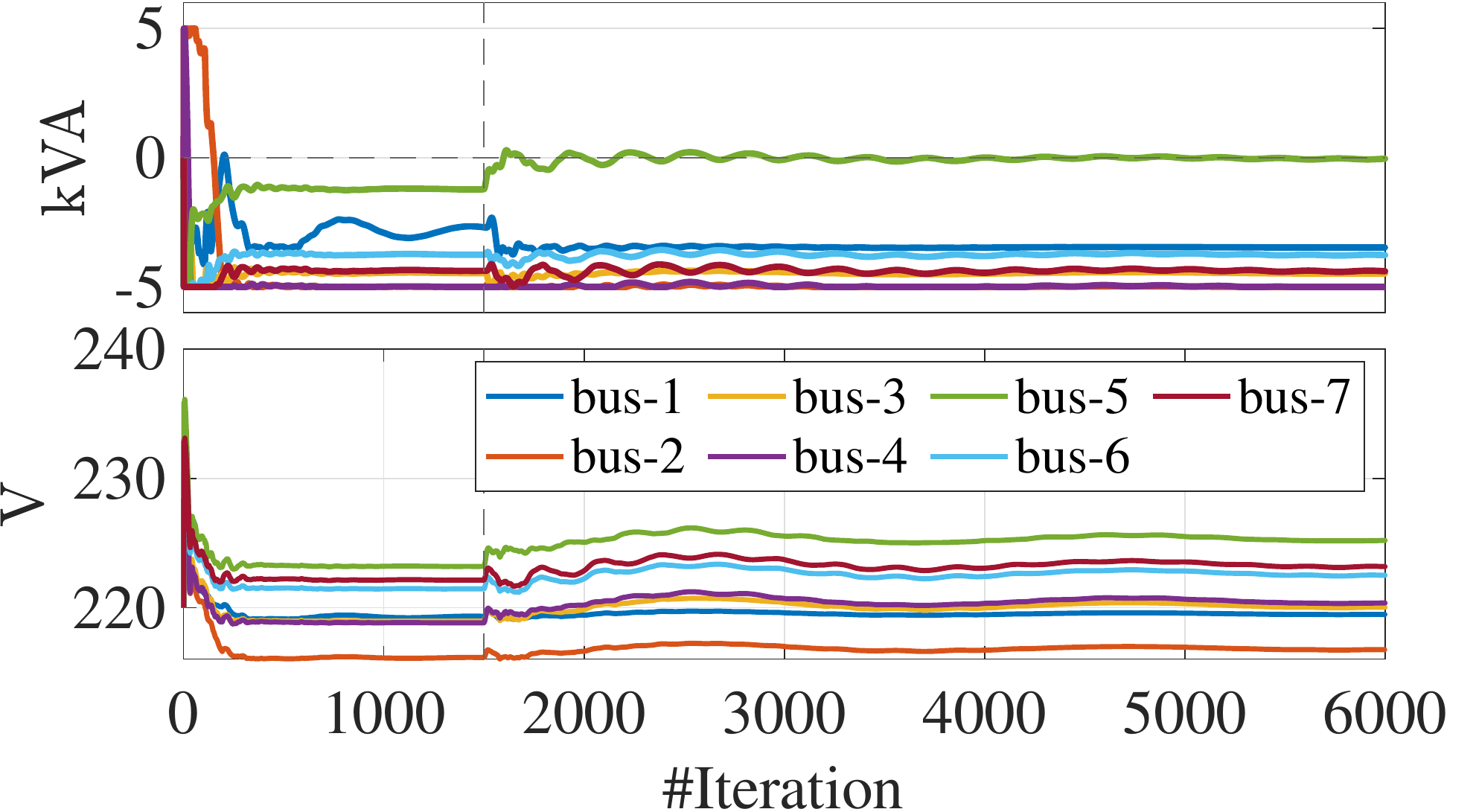}
\caption{Controller iterative convergence with attack. a) Reactive power; b) Voltage.}
\label{fig:aftera} 
\end{figure}
Let us suppose that node 5 has sufficient surplus PV power and a desire to feed as much as PV power to the grid for additional profit.
The node could become the attacker, then compromises the exchanging data and updates the data to its two-hop neighbor nodes 3 and 4, as illustrated in the data update pattern in Fig. \ref{fig:information}.
In the subsequent iterations, the entire network updates the decision variables within a contaminated environment. 
As shown in Fig. \ref{fig:aftera}, the attack happens at the 1500th iteration, where the vertical dash line lies.
The reactive power of the attacker is manipulated to be zero under the heuristic attack, indicating that more active PV power can be fed back to the grid by getting rid of the limitation of the reactive power priority. 
Simultaneously, the voltage fluctuations at all nodes lie within a feasible small region.
Therefore, such an adversarial behavior may help the attacker bypass the responsibility for providing or absorbing reactive power.

\section{Conclusion}
This paper considers the vulnerability of a decentralized inverter controller in a distribution network. 
We discuss the effect of the reactive power priority principle on the local two-hop information update and propose a heuristic attack.
The result demonstrates that the attacker can bypass the principle and provides more energy to the grid pursuing more profit.

In future work, the heuristic attack could be expanded to be a more complicated optimization method, and it would be desirable to develop a robust counterattack mechanism once the attack is detected.

\bibliographystyle{IEEEtran}
\bibliography{sutd1.bib}

\end{document}